\documentclass[11pt]{article}

\usepackage{amsmath,amssymb,amsfonts,amsthm,,a4}
\usepackage{url}
\usepackage{multirow}

\newtheorem{theorem}{Theorem}
\newtheorem{lemma}{Lemma}
\newtheorem{proposition}{Proposition}
\newtheorem{corollary}{Corollary}

\newtheorem{remark}{Remark}

\newtheorem*{hypothesis}{Hypothesis}

\usepackage{float}

\usepackage{hyperref}
\hypersetup{
    unicode=false,          
    pdftoolbar=true,        
    pdfmenubar=true,        
    pdffitwindow=true,      
    pdftitle={On the Certification of the Restricted Isometry Property},    
    pdfauthor={Pascal Koiran, Anastasios Zouzias},     
    pdfsubject={},   
    pdfcreator={Pascal Koiran},   
    pdfproducer={}, 
    pdfkeywords={compressed sensing, restricted isometry property, basis pursuit
}, 
    pdfnewwindow=true,      
    colorlinks=false,       
    linkcolor=red,          
    citecolor=green,        
    filecolor=magenta,      
    urlcolor=cyan          
}

\usepackage{algorithm}
\usepackage{algpseudocode}

\newcommand{\ignore}[1]{}

\newcommand{\norm}[1]{\ensuremath{\left\|#1\right\|}}

\newcommand{\Id}{\mathbf{I}}

\newcommand\cc{\ensuremath{\mathbb{C}}}

\title{Hidden Cliques
and the Certification of the Restricted Isometry Property}

\author{%
Pascal Koiran\thanks{LIP, UMR 5668, ENS de Lyon -- CNRS -- UCBL -- INRIA, 
\'Ecole Normale Sup\'erieure de Lyon, Universit\'e de Lyon. 
A part of this work was done when visiting the Department of Computer
Science at the University of Toronto. Email: {\tt Pascal.Koiran@ens-lyon.fr}.} 
\and 
Anastasios Zouzias\thanks{Department of Computer Science, 
University of Toronto. Email: {\tt zouzias@cs.toronto.edu}.}
}

\begin{document}

\maketitle

\begin{abstract}
Compressed sensing is a technique for finding sparse solutions to underdetermined linear systems. This technique relies on properties of the sensing matrix such as the \emph{restricted isometry property}. Sensing matrices that satisfy 
this property with optimal parameters are mainly obtained via probabilistic arguments. Deciding whether a given matrix satisfies the restricted isometry property is a non-trivial computational problem. Indeed, we 
show in this paper that restricted isometry parameters cannot be approximated 
in polynomial time within any constant factor under the assumption
that the hidden clique problem is hard.

 Moreover, on the positive side we propose an improvement on the brute-force enumeration algorithm for checking the restricted isometry property.
\end{abstract}

\section{Introduction}
%
Let $\Phi$ be a $n \times N$ matrix with $N \geq n$. 
A vector $x \in \cc^N$ is said
to be $k$-sparse if it has at most $k$ nonzero coordinates. 
Given $\delta \in ]0,1[$, $\phi$ is said to satisfy the Restricted Isometry
Property (RIP) of order $k$ with parameter $\delta$ if it approximately preserves the Euclidean norm in the following sense:
for every $k$-sparse vector $x$, we have
$$(1-\delta)||x||^2 \leq ||\Phi x||^2 \leq (1+\delta)||x||^2.$$
Clearly, for this to be possible we must have $k \leq n$.
Given $\delta$, $n$ and $N$, the goal is to construct RIP matrices with
$k$ as large as possible.
This problem is motivated by its applications to compressed sensing:
it is known from Cand\`es, Romberg and Tao~\cite{Candes,CRT06,CandesTao} that 
the restricted isometry property enables the efficient recovery of sparse signals using linear programming techniques.
For that purpose one can take any fixed $\delta < \sqrt{2}-1$~\cite{Candes}.

 Various probabilistic models are known to generate random matrices 
that satisfy the RIP with a value of  $k$ which is (almost) linear $n$.
See for instance Theorem~2 in Section~\ref{lazy} for the case of
matrices with entries that are independent symmetric $(\pm 1)$  
Bernouilli matrices. The recent survey~\cite{Vershynin} provides
additional results of this type and extensive references to the probabilistic
literature.
Some significant effort has been devoted to the construction of explicit
(rather than probabilistic) RIP matrices, but this appears to be a 
difficult problem. As pointed out by Bourgain et al. in 
a recent paper~\cite{BDFKK,BDFKKb}, most of the known explicit
constructions~\cite{Kashin75,AGHP92,Devore07}
are based on the construction of systems of unit vectors
with a small coherence parameter (see section~\ref{order} for a definition of this parameter and its connection to the RIP).
Unfortunately, this method cannot produce 
RIP matrices of order $k > \sqrt{n}$~\cite{BDFKK,BDFKKb}.
Bourgain et al. still manage to break through the $\sqrt{n}$
``barrier'' using techniques from additive combinatorics: they construct
RIP matrices of order $k=n^{1/2+\epsilon_0}$ where $\epsilon_0 > 0$ is an
unspecified ``explicit constant''. Note that this is still far from the 
order achieved by probabilistic constructions.
Here we study the restricted isometry property from the point
of view of computational complexity: what is the complexity of deciding
whether a matrix satisfies the RIP, and of computing or approximating
its order $k$ or its RIP parameter $\delta$? 
An efficient (deterministic) algorithm would have applications
to the construction of RIP matrices. One would draw a random matrix $\Phi$
from one of the well-established probabilistic models mentioned above,
and run this hypothetical algorithm on $\Phi$ to compute or approximate
$k$ and $\delta$. The result would be a matrix with {\em certified}
restricted isometry properties (see Section~\ref{lazy} for an actual result
along those lines). This may be the next best thing short of
an explicit construction (and as mentioned above, the known  explicit constructions are far from optimal).
The definition of the restricted isometry property suggests an exhaustive
search over $\binom{N}{k}$ subspaces, but prior to this work there was little 
evidence that checking the RIP is computationally hard (more on this in Section~\ref{hard_intro}).
There has been more work from the algorithm design side. 
In particular, it was shown that semi-definite programming can be used
to verify the restricted isometry property~\cite{ABG08} and other related properties 
from compressed sensing~\cite{AspreG08,IouNem08}. Unfortunately, 
as pointed out in~\cite{AspreG08} these methods are unable to certify
the restricted isometry property
 for $k$ larger than $O(\sqrt{n})$, even for matrices that satisfy the RIP up to order $\Omega(n)$. As we have seen,
$k=O(\sqrt{n})$ is also the range where coherence-based methods reach their limits.
In this paper we provide both positive and negative results on the computational complexity of the RIP, 
including the range $k > \sqrt{n}$.

%
\subsection{Positive Results}
%
In Section~\ref{order}, we study the relation between the RIP parameters
of different orders for a given matrix $\Phi$. Very roughly, we show
in Theorem~\ref{m2k} 
that the RIP parameter is at most proportional to the order.
We therefore have a trade-off between order and RIP parameter:
in order to construct a matrix of given order and RIP parameter, it suffices
to construct a matrix of lower order and smaller RIP parameter.
We illustrate this point in Section~\ref{lazy}. Our starting point
is the above-mentioned (very naive) exhaustive search algorithm,
which enumerates all $\binom{N}{k}$ subspaces generated by $k$ 
column vectors. 
We obtain  a ``lazy algorithm'' which enumerates instead all subspaces
generated by $l$ basis vectors for some $l < k$.
We show that the lazy algorithm can go slightly beyond the $\sqrt{n}$
barrier if a quasi-polynomial running time is allowed.
%
\subsection{Negative Results: the Connection to Hidden Cliques} \label{hard_intro}
%

We show that RIP parameters are hard to approximate  within any constant
factor under the assumption that the hidden clique problem is hard.
In fact, we need an assumption (spelled 
out at the end of this subsection) which is somewhat weaker than the usual one.
Our hardness result applies to any order of the form $k=n^{\alpha}$, 
where $\alpha$ is any constant  in the interval $]0,1[$.
It applies to square as well as to rectangular matrices.
We gave similar results in the unpublished manuscript~\cite{KZ11} under
a (nonstandard) assumption on the complexity of detecting dense subgraphs.
By contrast, as explained below the hypothesis that we use in this paper is 
well established.
Prior to our work, little was known on the hardness of checking the
restricted isometry property. It was pointed out by Terence Tao~\cite{Tao07} 
that ``there is no fast (e.g. sub-exponential time) algorithm known to test whether any given matrix is UUP or not.''\footnote{In his blog post, Tao uses
the notation ``UUP'' for the RIP.}
As to  hardness results, one can mention the NP-hardness proof of~\cite{BDMS12}, which is based on the following (known) fact: it is NP-hard 
to distinguish a matrix with a nonzero $k$-sparse 
vector in its kernel from a matrix without any such vector in its kernel.
In the first case, the matrix does not satisfy the RIP of order $k$, while
in the second case it does satisfy the RIP of order $k$ for some parameter
$\delta$. Since $\delta$ may be very close to 1, this result does not say much on the complexity of approximating the RIP parameters. A similar result
was obtained in~\cite{PT12}.

The size of the largest clique in a typical graph drawn from the $G(n,1/2)$ 
distribution is roughly 
$2 \log_2 n$.
In the hidden clique problem, one must find a clique of
size  $t \gg 2 \log_2 n$ which was planted at random in a random graph.
This problem is solvable in polynomial time for a clique of size 
$t=\Theta(\sqrt{n})$~\cite{AKS98}.
It is widely believed, however,
 that the problem cannot be solved in polynomial time
 for a planted clique of size $t=n^c$, where $c$ is any constant in 
the open interval $]0,1/2[$.
Even the more modest goal of distinguishing between a random graph
and a random graph with a planted clique of size $n^c$ is believed
to require more than polynomial time~\cite{AAMMW11} 
(see appendix B.4 of~\cite{AAK07} for a comparison of  distinguishing
versus finding hidden cliques).

In the last few years, several hardness results have been obtained under
the assumption that the hidden clique problem is not polynomial time 
solvable~\cite{AAMMW11,AAK07,HK11}. 
We refer to~\cite{AAMMW11} for more information on the history of 
this problem.

In this paper, we show hardness of approximation for RIP parameters
under the following assumption. We actually have a family of assumptions,
parameterized by the clique size (in keeping with the tradition in this area~\cite{AKS98}, we omit floor and ceiling signs  to simplify the presentation).
\begin{hypothesis}[$\mathbf H_{\epsilon}$]
There is no polynomial time algorithm $\cal A$ 
which, given as input a graph $G$
on $n$ vertices:
\begin{itemize}
\item always outputs ``yes'' 
if $G$ contains a clique of size $n^{\frac{1}{2}-\epsilon}$.

\item Outputs ``no clique'' on most graphs $G$ when $G$ is drawn from the uniform
distribution $G(n,1/2)$.
\end{itemize}
\end{hypothesis}
In other words, $(H_{\epsilon})$ asserts that no polynomial time algorithm can
certify the absence of a clique of size $n^{\frac{1}{2}-\epsilon}$ 
from most graphs on $n$ vertices (where ``most graphs'' means: with probability
approaching 1 as $n \rightarrow + \infty$).
Note that this is a one-sided hypothesis: algorithm $\cal A$ is allowed to 
 err (rarely) but only on input graphs 
that do not contain a a clique of size 
$n^{\frac{1}{2}-\epsilon}$.

Note also that Hypothesis $H_{\epsilon}$ becomes increasingly stronger as 
$\epsilon \rightarrow 0$ (and it becomes false 
for $\epsilon<0$: if $\alpha>1/2$, 
a simple spectral algorithm can certify that most graphs
on $n$ vertices do not contain any clique of size $n^{\alpha}$. For completeness, 
we give a proof in the appendix).
Hypothesis $H_{\epsilon}$ is clearly true if it is hard to distinguish beween a random graph and a random graph with a planted clique of size $n^{\frac{1}{2}-\epsilon}$. It is therefore consistent with current knowledge to assume
that  $(H_{\epsilon})$ holds true for all 
constants $\epsilon \in ]0,1/2[$.

\subsection{Organization of the Paper} 
%
As explained above, the next two sections are devoted to positive results.
In Section~\ref{eigenvalues} we work out some bounds on the eigenvalues of
random matrices, for later use in our reductions from hidden clique
to the approximation of RIP parameters.
We rely mainly on the classical work of F\"uredi and Koml\'os~\cite{FK81}
as well as on a more recent concentration inequality due to Alon, Krivelevich
and Vu~\cite{AKV02}.
In Section~\ref{hard_square} we use these eigenvalue bound to show that 
approximating RIP parameters is hard even for square matrices.
In Section~\ref{hard_rect} we derive similar results for matrices of
``strictly rectangular'' format (which is the case of interest in compressed
sensing).
We proceed by reduction from the square case.
Interestingly, this last reduction relies on the known constructions
(deterministic~\cite{BDFKK,BDFKKb} and probabilistic~\cite{Vershynin}) 
of matrices with good RIP parameters
mentioned earlier in the introduction.
We therefore turn these positive results into negative results.
The table at the end of Section~\ref{hard_rect} gives a summary of 
our hardness results. 

\section{Increasing the Order by Decreasing the RIP Parameter}
\label{order}
As explained at the beginning of~\cite{BDFKK,BDFKKb}, 
certain (suboptimal) constructions
are based on the construction of systems of unit vectors 
$(u_1,\ldots,u_N) \in \cc^n$ with small coherence. The coherence parameter
$\mu$ is defined as $\max_{i \neq j} |\langle u_i, u_j \rangle|$.
Indeed, we have the following proposition.
\begin{proposition} \label{coherence}
Assume that the column vectors $u_1,\ldots,u_N$ of $\Phi$ are of norm~1
and coherence $\mu$. Then $\Phi$ satisfies the RIP of order $k$ with
parameter $\delta=(k-1)\mu$.
\end{proposition}
We reproduce the proof from~\cite{BDFKK,BDFKKb} since if fits in one line:
for any $k$-sparse vector $x$, 
$$| ||\Phi x||^2 - ||x||^2| \leq 2 \sum_{i<j} |x_ix_j \langle u_i,u_j \rangle| \leq \mu ((\sum_i |x_i|)^2 - ||x||^2) \leq (k-1)\mu||x||^2.$$
%

%
We now give a result, which (as we shall see) generalizes Proposition~\ref{coherence}.
\begin{theorem} \label{m2k}
Assume that $\Phi$ has unit column vectors and satisfies the RIP of order $m$ with parameter $\epsilon$. For $k \geq m$, $\Phi$ also satisfies the RIP
of order $k$ with parameter $\delta = \epsilon(k-1)/(m-1)$.
\end{theorem}
\begin{proof}
Let $u_1,\ldots,u_N$ be the column vectors of $\Phi$. Let $x$ be a $k$-sparse
vector, and 
write $x=\sum_{i \in T} x_iu_i$ where $T$ is a subset of $\{1,\ldots,N\}$ of
size $k$.
Since $||\Phi x||^2 = ||x||^2 + 2 \sum_{i<j} x_ix_j \langle u_i,u_j \rangle$, to check the RIP of order $k$ we need to show that
\begin{equation} \label{innerprods}
|\sum_{i<j} x_i x_j \langle u_i,u_j \rangle| \leq \delta||x||^2/2,
\end{equation}
where $\delta = \epsilon(k-1)/(m-1)$.
To estimate the left hand side, we compare it to the sum of 
the similar quantity taken over all subsets of size $m$ of $T$, namely:
\begin{equation} \label{sum}
|\sum_{|S|=m} \sum_{i,j \in S, i<j} x_i x_j \langle u_i,u_j \rangle|.
\end{equation}
Since each pair $(i,j)$ appears in exactly $\binom{k-2}{m-2}$ subsets of 
size $m$, this sum is equal to $\binom{k-2}{m-2}$ times the left-hand side
of~(\ref{innerprods}). But we can also estimate (\ref{sum}) using the RIP
of order $m$. 
For each subset $S$ of size $m$, we have 
$$|\sum_{i,j \in S, i<j} x_i x_j \langle u_i,u_j \rangle| \leq \epsilon\sum_{i \in S}x_i^2/2.$$
This follows from~(\ref{innerprods}), replacing $\delta$ 
by $\epsilon$ (the RIP parameter of order $m$).
Since each term $x_i^2$ will appear exactly in $\binom{k-1}{m-1}$ subsets,
we obtain $\epsilon \binom{k-1}{m-1}||x||^2/2$ 
as an upper bound for~(\ref{sum}).
We conclude that  the left-hand side
of~(\ref{innerprods}) is bounded by $\frac{\epsilon}{2} \binom{k-1}{m-1}||x||^2/\binom{k-2}{m-2} = 
\epsilon \frac{k-1}{m-1}||x||^2/2$.\end{proof}
We claim that Proposition~\ref{coherence} is the case $m=2$ of 
Theorem~\ref{m2k}. This follows from the following observation.
\begin{remark} \label{order2}
For a matrix $\Phi$ with unit column vectors, the coherence parameter $\mu$
is equal to the RIP parameter of order 2.
\end{remark}
\begin{proof}
Let $\delta$ be the RIP parameter of order 2. We have $\delta \leq \mu$ by Proposition~\ref{coherence}. It remains to show that $\delta \geq \mu$.
Consider therefore two column vectors $u_i$ and $u_j$ with 
$|\langle u_i,u_j \rangle| = \mu$. 
Let $x=u_i+u_j$. We have $||x||^2=2$ and $||\Phi x||^2=2 \pm 2\mu$, so 
that $\delta \geq \mu$ indeed.
\end{proof}

\section{A Matrix Certification Algorithm} \label{lazy}
%
%


%
%
The naive algorithm for computing the RIP parameter of order $k$ will involve the enumeration of the $\binom{N}{k}$ submatrices of $\Phi$ made up
of $k$ column vectors of $\Phi$. For each $T \subseteq \{1,\ldots,N\}$ of
size $k$ let us denote by $\Phi_T$ the corresponding $n \times k$ matrix.
We need to compute (or upper bound) $\delta = \max_T \delta_T$, where 
$$\delta_T = \sup_{x \in \cc^k} |\ ||\Phi_T x||^2 /||x||^2 -1\ |.$$
For each $T$, $\delta_T$ can be computed efficiently by linear algebra. 
For instance, $\delta_T$ is the spectral radius of the self-adjoint matrix
$\Phi_T^* \Phi_T - \Id_k$.
The cost of the computation is therefore dominated by the combinatorial factor
$\binom{N}{k}$ due to the enumeration of all subsets of size $k$.
Here we analyze what the naive algorithm can gain from Theorem~\ref{m2k}.
We therefore consider the following {\em lazy algorithm.}
\begin{algorithm}{}
	\caption{}\label{alg:lazy}
\begin{algorithmic}[1]
\Procedure{Lazy}{$\Phi$, $m$, $\delta$}
\State {\bf Input:} a $n \times N$ matrix $\Phi$ with unit column vectors, an integer $m \leq n$,
and a parameter $\delta \in ]0,1[$.
\State Compute as explained above the RIP parameter of order $m$. Call it $\epsilon$.
\State {\bf Output:} Certify $\Phi$ as a RIP matrix of order $k$ with parameter $\delta$,
for all $k \geq m$ such that $\epsilon (k-1)/(m-1) \leq \delta$.
\EndProcedure 
\end{algorithmic}
\end{algorithm}
The correctness of the 
algorithm follows immediately from Theorem~\ref{m2k}. We now analyze its behavior on random matrices, which are in many cases known to satisfy the RIP with high probability. Consider for instance 
the case of a matrix whose entries are independent symmetric Bernouilli
random variables.
\begin{theorem} \label{randrip}
Let $A$ be a $n \times N$ matrix whose entries are  independent symmetric 
Bernouilli random variables and assume that 
$n \geq C \epsilon^{-2}m \log(eN/m)$.
With probability at least $1-2\exp(-c\epsilon^2n)$, 
the normalized matrix $\Phi=\frac1{\sqrt{n}}A$ 
satisfies the RIP of order $m$ with parameter $\epsilon$.
Here $C$ and $c$ are absolute constants.
\end{theorem}
In fact the same theorem holds for a very large class of random matrix models,
namely, subgaussian matrices with either independent rows or independent columns (\cite{Vershynin}, Theorem~64).
\begin{proposition} \label{lazy_analysis}
Let $A$ be a random matrix as in Theorem~\ref{randrip}, 
and $\delta \in ]0,1[$. 
With probability at least $1-2(eN/m)^{-cCm}$, 
the lazy algorithm presented above will certify 
that $A$ satisfies the RIP of order $k$ with parameter $\delta$ for
all $k$ such that: $$k \leq \delta \sqrt{\frac{mn}{c \log (eN/m)}}.$$
Here $c$ and $C$ are the absolute constants from Theorem~\ref{randrip}.
\end{proposition}
\begin{proof}
All parameters being fixed we take $\epsilon$ as small as allowed by 
Theorem~\ref{randrip}, so that $n\epsilon^2=Cm \log(eN/m)$.
This yields the announced probability estimate, and the upper bound on $k$ 
is $\delta m/\epsilon$.
\end{proof}
To compare the lazy algorithm to the naive algorithm, set for instance 
$m=\sqrt{n}$. In applications to compressed sensing one can set $\delta$ to
a small constant value (any $\delta < \sqrt{2}-1$ will do). Thus, disregarding constant and
logarithmic factors, with high probability the lazy algorithm will certify 
the RIP property for $k$ of order roughly $n^{3/4}$.
This is achieved by enumerating 
$\binom{N}{n^{1/2}}$ subspaces, whereas the naive algorithm would enumerate roughly $\binom{N}{n^{3/4}}$ subspaces.

Another choice of parameters in Proposition~\ref{lazy_analysis} shows that one can beat the $\sqrt{n}$ 
bound by a logarithmic factor
with a quasi-polynomial time algorithm.
For instance:
\begin{corollary}
If we set $m=(\log N)^3$, the lazy algorithm runs in time $2^{O(\log^4 N)}$
and, with probability at least $1-2^{-\Omega((\log N)^4)}$
 certifies that $A$ satisfies the RIP of order $k$ with parameter $\delta$
for all $k \leq K \delta \log N \sqrt{n}$, where $K$ is an absolute constant.
\end{corollary}
%
%

\section{Eigenvalues of Random Symmetric Matrices} \label{eigenvalues}
%
Proposition~\ref{modelC} is the main probabilistic inequality that we derive in this section.
It shows that square matrices obtained by Cholesky decomposition from a certain class of random matrices
have good RIP parameters with high probability. This result is then used in Section~\ref{hard_square} 
to give a reduction from  hidden clique 
to the approximation of RIP parameters.
%
\subsection{Model A}
%
Consider the following random matrix model: $A$ is a symmetric $k \times k$ 
matrix with $a_{ii}=0$, and for $i < j$ the $a_{ij}$ are independent symmetric
Bernouilli random variables.

Let $\lambda_1(A) \geq \lambda_2(A) \geq \dots \lambda_k(A)$ 
be the eigenvalues of $A$. Let $m_s$ be the median of $\lambda_s(A)$.
From the main result of~\cite{AKV02} (bottom of p.~$263$) we have for $t \geq 0$
the inequality:
\[\Pr[\lambda_s(A) - m_s \geq t] \leq 2e^{-t^2/32s^2}.\]
From F\"uredi and Koml\'os (\cite{FK81}, Theorem~2) we know that 
$m_1 \leq 3 \sigma \sqrt{k}$ for~$k$ large enough, 
where $\sigma=1$ is the standard deviation of the $a_{ij}$ in the case $i<j$.
Therefore we have 
\[\Pr[\lambda_1(A) \geq 3 \sqrt{k}+t] \leq 2e^{-t^2/32}.\]
Since $\lambda_k(A)=-\lambda_1(-A)$ and $-A$ has same distribution as $A$,
we also have 
\[\Pr[\lambda_k(A) \leq -3 \sqrt{k}-t] \leq 2e^{-t^2/32}\]
(one could also apply directly the bound on $\lambda_k(A)$ for the more
general model considered in~\cite{AKV02}).
As a result:
\begin{proposition} \label{modelA}
There is an integer $k_0$ such that for all $k \geq k_0$ 
and for all $t \geq 0$ we have:
\[\Pr[ \max_i |\lambda_i(A)| \geq 3 \sqrt{k}+t] \leq 4e^{-t^2/32}.\]
\end{proposition}
\begin{remark} \label{constant3}
The constant 3 in Proposition~\ref{modelA} can be replaced
by any constant bigger than 2 (see Theorem~2 in~\cite{FK81}).
\end{remark}
%
\subsection{Model B}
%
Next we consider the model where $B$ is a symmetric $k \times k$ matrix satisfying the following condition: $b_{ii}=1$, and $b_{ij}=c \cdot a_{ij}/\sqrt{n}$ 
 for $i<j$, where the $a_{ij}$ are independent symmetric
Bernouilli random variables. Here $c>0$ is a fixed constant, and $n$ is an additional parameter 
which should be thought of as going to infinity with $k$.
\begin{corollary} \label{modelB}
Assume that $k \geq k_0$ 
and that 
$\delta\sqrt{n} \geq 3c\sqrt{k}$.
Then the eigenvalues of $B$ all lie in the interval
$[1-\delta,1+\delta]$ with probability at least
\[1-4\exp[-(\frac{\delta\sqrt{n}}{c} -3\sqrt{k})^2/32].\]
\end{corollary}
\begin{proof}
We have $B=\Id_k + cA/\sqrt{n}$, where $A$ follows the model of Proposition~\ref{modelA}. The result therefore follows from that proposition by choosing
$t$ so that $c(3\sqrt{k}+t)/\sqrt{n}=\delta$, i.e., 
$t=\delta\sqrt{n}/c-3\sqrt{k}$.
\end{proof}
In the next corollary we look at the case $n=k$ of this model.
\begin{corollary} \label{modelB'}
Assume  that $n \geq k_0$ 
and  $3c < 1$.
Then $B$ is positive semi-definite with probability at least
$$1-4\exp[-(1/c -3)^2 n/32].$$
\end{corollary} 
\begin{proof}
Set $n=k$ and $\delta=1$ in Corollary~\ref{modelB}.
\end{proof}
In the last result of this subsection we consider again the 
model $B=\Id_n + cA/\sqrt{n}$. 
Given a $n \times n$ matrix $M$ 
and two subsets $S,T \subseteq \{1,\ldots,n\}$ of size $k$, 
let us denote by $M_{S,T}$ the $k \times k$ sub-matrix made up of all entries
of $M$ of row number in $S$ and column number in $T$.
\begin{corollary} \label{modelB''}
Consider the random matrix 
$B=\Id_n + cA/\sqrt{n}$ where $A$ is drawn from
the uniform distribution on the set $n \times n$ symmetric matrices
with null diagonal entries and $\pm 1$ off-diagonal entries.

If $n \geq k \geq k_0$, then with probability at least 
 $$1-4\exp\left[k \ln (ne/k)-( \frac{\delta\sqrt{n}}{c} 
-3\sqrt{k})^2/32\right]$$
the submatrices $B_{S,S}$ have all their eigenvalues in the interval 
$[1-\delta,1+\delta]$ for all subsets $S\subseteq \{1,\ldots,n\}$ of size $k$.
\end{corollary}
\begin{proof}
By Corollary~\ref{modelB}, for each fixed $S$ 
matrix $B_{S,S}$ has an eigenvalue outside of the interval 
$[1-\delta,1+\delta]$ with probability at most
$4\exp[-( \frac{\delta\sqrt{n}}{c} -3\sqrt{k})^2/32].$
The result follows by taking a union bound over the 
$\binom{n}{k} \leq (ne/k)^k$ subsets of size $k$.
\end{proof}
%
\subsection{Model C} \label{sectionC}
%
In Corollaries~\ref{modelB'} and~\ref{modelB''}
 we considered the following random model for $B$:
set $B = \Id_n + c A/\sqrt{n}$, where $A$ is chosen from the uniform distribution on the set $S_n$ of all symmetric matrices with null diagonal entries
and $\pm 1$ off-diagonal entries.
If $B$ is positive semi-definite, we can find by Cholesky decomposition 
a $n \times n$
matrix $C$ such that $C^TC=B$.
If $B$ is not positive semi-definite, we set $C=0$.
This is the random model for $C$ that we study in this subsection.
\begin{proposition} \label{modelC}
Assume that $n \geq k \geq k_0$ 
and that 
$3c < \min(1, \delta\sqrt{n}/\sqrt{k})$.
With probability at least
 $$1-4\exp\left[k \ln (ne/k)-(\frac{\delta\sqrt{n}}{c} 
-3\sqrt{k})^2/32\right]-4\exp[-(1/c -3)^2 n/32],$$
$C$ satisfies the RIP of order $k$ with parameter $\delta$.
\end{proposition}
\begin{proof}
If $B = \Id_n + c A/\sqrt{n}$ is not positive semi-definite then $C=0$
and this matrix obviously does not satisfy the RIP.
By Corollary~\ref{modelB'}, $B$ can fail to be positive semi-definite
with probability at most $4\exp[-(1/c -3)^2 n/32].$
If $B$ is positive semi-definite then $C^TC=B$.
Using the notation of Corollary~\ref{modelB''}, 
matrix $C$ satisfies the RIP of order $k$ with parameter $\delta$ 
if for all subsets $S$ of size $k$,
the eigenvalues of the $k \times k$ matrices $(C^TC)_{S,S}$
all lie in the interval $[1-\delta,1+\delta]$.
Since $C^TC=B$, by Corollary~\ref{modelB''} this can happen with probability 
at most $4\exp[k \ln (ne/k)-( \frac{\delta\sqrt{n}}{c} 
-3\sqrt{k})^2/32].$
\end{proof}

\section{Large Cliques and the Restricted Isometry Property}
\label{hard_square}
%
In this section we show 
(in Theorems~\ref{1/3th}, \ref{point6th} 
and more generally in Theorem~\ref{highk}) 
that RIP parameters are hard to approximate even for square matrices. 
We establish connections between hidden clique problems
and the RIP thanks to a generic reduction 
which we call the {\em Cholesky reduction}.
This reduction maps a graph $G$ on $n$ vertices to a $n \times n$ matrix
$C(G)$. Let $A$ be the signed adjacency matrix of $G$: we have
$a_{ii}=0$ and for $i \neq j$, $a_{ij}=1$ if $ij \in E$; 
$a_{ij}=-1$ if $ij {\not \in} E$.
We construct $C=C(G)$ from $A$ 
using the procedure described in Section~\ref{sectionC}.
That is, we first compute $B = \Id_n + c A/\sqrt{n}$.
Here $c$ is some absolute constant smaller for $1/3$, for instance 
$c=0.3$.
If $B$ is not positive semi-definite, we set $C=0$. Otherwise, we find
by Cholesky decomposition a matrix $C$ such that $C^TC=B$.

For suitable values of $k$, $C(G)$ satisfies the RIP of order $k$ for 
most graphs $G$. This was made precise in Proposition~\ref{modelC}.
On the other hand, if $G$ has a $k$-clique then 
$C(G)$ cannot satisfy the RIP of order $k$ for a small value of the
parameter $\delta$. In order to show this, we first need a simple lemma.
\begin{lemma} \label{expansion}
Let $A$ be the signed adjacency matrix of a graph $G$.
If $G$ has a clique of size $k$ 
then there is a unit vector supported by $k$ basis vectors such that 
$x^T A x =  k-1$.
\end{lemma}

\begin{proof}
Let $H$ be the $k$-clique.
Here is a suitable vector: 
set $x_i =1/\sqrt{k}$ if $i \in H$ and $x_i=0$ otherwise.
\end{proof}

\begin{proposition} \label{clique}
If $G$ has a clique of size $k$ and $\delta < c(k-1)/\sqrt{n}$
then
$C(G)$ does not satisfy the RIP of order $k$ with parameter $\delta$.
\end{proposition}

\begin{proof}
If $B$ is not semi-definite positive, $C(G)=0$ does not satisfy the RIP.
Otherwise $C^TC=B$. Let $x$ be the vector of Lemma~\ref{expansion}.
We have $||Cx||^2=x^TC^TCx=x^TBx=1+cx^TAx/\sqrt{n} > 1+\delta$.
\end{proof}


We can now prove our first hardness results.
We first illustrate our method on two examples, and then prove a general
result at the end of this section.

\begin{theorem} \label{1/3th}
Assume hypothesis $(H_{1/6})$, that is: no polynomial time algorithm 
can certify that most graphs do not contain a clique of size $n^{1/3}$.
Then, no polynomial time algorithm
can distinguish a matrix with RIP parameter of order $n^{1/3}$ at most $n^{-1/4}$
from a matrix with RIP parameter of order $n^{1/3}$ at least $n^{-1/6}/4$.
\end{theorem}
\begin{proof}
We show the contrapositive: 
assuming the existing of a distinguishing algorithm $\cal A$, 
we construct an algorithm that contradicts hypothesis $(H_{1/6})$.
Fix a constant $c<1/3$, for instance $c=0.3$. 
On input $G$, this algorithm first contruct $C(G)$.

If $G$ contains a clique of size $k=n^{1/3}$ 
 then by Proposition~\ref{clique} the matrix $C(G)$
does not satisfy the RIP of order $k$ with parameter $c'n^{-1/6}$.
Here $c'<c$ is another constant (for $n$ large enough we can take $c'=1/4$).

We consider now the case where $G$ was drawn from the $G(n,1/2)$ distribution.
Set $\delta=n^{-1/4}$. We can apply Proposition~\ref{modelC} 
since $\delta \sqrt{n} /\sqrt{k}=n^{1/12} > 1 > 3c$.
This proposition shows that with probability approaching 1 as 
$n \rightarrow + \infty$, $C(G)$ satisfies the RIP of order $k$ with parameter
$\delta$. 

We can therefore call algorithm $\cal A$ to 
certify the absence of a clique of size~$n^{1/3}$.
More precisely, if $G$ contains a $k$-clique our algorithm always finds out.
On the other hand, if $G$ was drawn  from $G(n,1/2)$ our algorithm 
answers correctly with high probability.
\end{proof}

This theorem implies in particular than RIP parameters cannot be approximated
 within any constant factor.
We can obtain a similar result for an order $k>\sqrt{n}$ under the
same hypothesis. 
This is possible essentially because a matrix that doesn't satisfy 
the RIP for a given order $k$ cannot satisfy the RIP for any order $k'>k$.
\begin{theorem} \label{point6th}
Assume Hypothesis~$(H_{1/6})$ as in the previous theorem.
Then no polynomial time algorithm
can distinguish a matrix with RIP parameter of order $n^{0.6}$ 
at most $n^{-0.19}$
from a matrix with RIP parameter of order $n^{0.6}$ at least $n^{-1/6}/4$.
\end{theorem}
\begin{proof}
We proceed as in the proof of the previous theorem: 
assuming the existing of a distinguishing algorithm $\cal A$, 
we construct an algorithm that contradicts the hypothesis.

If $G$ contains a clique of size $n^{1/3}$ 
then we saw that $C(G)$
does not satisfy the RIP of order $n^{1/3}$ with parameter $n^{-1/6}/4$.
It is {\em a fortiori} the case that this matrix does not satisfy 
the RIP of order $k=n^{0.6}> n^{1/3}$ with parameter $n^{-1/6}/4$.

We consider now the case where $G$ was drawn from the $G(n,1/2)$ distribution.
Set $\delta=n^{-0.19}$. We can apply Proposition~\ref{modelC} 
since $\delta \sqrt{n} /\sqrt{k}= n^{0.01}>1 > 3c$.
Consider the argument of the first exponential
in the probability bound of Proposition~\ref{modelC}.
The positive term $k \ln (ne/k)$, which is of order $n^{0.6}\ln n$,
is dominated by the negative term 
$(\frac{\delta\sqrt{n}}{c} -3\sqrt{k})^2/32$,
which is of order $n^{0.62}$.
We conclude that with probability approaching 1 as 
$n \rightarrow + \infty$, $C(G)$ satisfies the RIP of order $k$ with parameter
$\delta$. 

We can therefore call algorithm $\cal A$ to 
certify the absence of a clique of size~$n^{1/3}$.
More precisely, if $G$ contains a clique of size $n^{1/3}$ 
our algorithm always finds out.
On the other hand, if $G$ was drawn  from $G(n,1/2)$ our algorithm 
answers correctly with high probability.
\end{proof}
More generally, we have the following result.
\begin{theorem} \label{highk}
Set $k=n^{(1-2\epsilon)(1-\epsilon)}$ where $\epsilon \in ]0,1/2[$.
Set also $\delta=n^{-5\epsilon/4+\epsilon^2/2}$.

Hypothesis $(H_{\epsilon})$ implies that 
no polynomial time algorithm
can distinguish a matrix with RIP parameter of order $k$ 
at most $\delta$
from a matrix with RIP parameter of order $k$ 
at least $n^{-\epsilon}/4$.

In particular, since $\delta=o(n^{-\epsilon}/4)$, it follows that 
no polynomial time algorithm can approximate the RIP parameter of order $k$
within any constant factor.
\end{theorem}

\begin{remark}
The exponent $\alpha=(1-2\epsilon)(1-\epsilon)$ ranges over $]0,1[$ 
as $\epsilon$ ranges over the interval $]0,1/2[$.
This theorem therefore shows that for {\em any} exponent $\alpha \in ]0,1[$,
the RIP parameter of order $k=n^{\alpha}$ cannot be approximated 
within any constant factor in polynomial time.
\end{remark}

\begin{proof}[Proof of Theorem~\ref{highk}]
That $\delta=o(n^{-\epsilon}/4)$ follows from the inequality
$-5\epsilon/4+\epsilon^2/2 < -\epsilon/4$. This inequality 
holds true for all
$\epsilon \in ]0,2[$, and in particular for all $\epsilon$ in the range
$]0,1/2[$ that is of interest here.

We now prove the main part  of the theorem.
Assuming the existence of a distinguishing algorithm $\cal A$, 
we construct again an algorithm that refutes hypothesis $(H_{\epsilon})$.

We set as usual $c=0.3$.
If $G$ contains a clique of size $n^{1/2-{\epsilon}}$ 
then by Proposition~\ref{clique} $C(G)$ does not satisfy the RIP of order 
$n^{1/2-{\epsilon}}$ with parameter $n^{-\epsilon}/4$.
It is {\em a fortiori} the case that this matrix does not satisfy 
the RIP of order $k=n^{(1-2\epsilon)(1-\epsilon)}> n^{(1-2\epsilon)/2}$ 
for the same parameter value.

Consider now the case where $G$ is drawn from the $G(n,1/2)$ distribution.
We can apply Proposition~\ref{modelC} 
since $\delta \sqrt{n} /\sqrt{k}= n^{\frac{\epsilon}{4}(1-2\epsilon)}>1 > 3c$.
Consider the argument of the first exponential term
in the probability bound of Proposition~\ref{modelC}.
The positive term $k \ln (ne/k)$, which is of order 
$k \ln n = n^{(1-2\epsilon)(1-\epsilon)} \ln n$, 
is dominated by the negative term 
$(\frac{\delta\sqrt{n}}{c} -3\sqrt{k})^2/32$, which is of order 
$\delta^2 n = n^{1-5\epsilon/2+\epsilon^2}$.
Indeed, the difference in the two exponents is 
$$1-\frac{5\epsilon}{2}+\epsilon^2-(1-2\epsilon)(1-\epsilon)=
\frac{\epsilon}{2}-\epsilon^2>0.$$
As a result, with probability approaching 1 as $n \rightarrow +\infty$, 
$C(G)$ satisfies the RIP of order $k$ with parameter $\delta$.
We can therefore refute hypothesis $(H_{\epsilon})$ by running algorithm $\cal A$ on input $C(G)$.
\end{proof}

\section{Hardness for Rectangular Matrices}
\label{hard_rect}

In this section we show that the RIP parameters of rectangular
matrices are hard to approximate. This is the case of interest
in compressed sensing. 
In a sense this was already done in Section~\ref{hard_square}:
we have shown that the special case of square matrices is already hard.
Nevertheless, it is of interest to know that the problems remains hard
for {\em strictly rectangular} matrices. This is what we do in this 
section. Proofs are essentially by reduction from the square case.
We begin with a simple lemma.
\begin{lemma} \label{diag_rip}
Consider a matrix $\Phi$ with the 
block structure
$$\Phi=\left(
\begin{array}{cc}
A & 0\\  0 & B
\end{array}\right),$$
where $A$ and $B$ both have at least $k$ columns.
This matrix satisfies the RIP of order $k$ with parameter $\delta$ 
if and only if the same is true for both $A$ and $B$.
\end{lemma}
\begin{proof}
For an input vector $x$ with the corresponding block structure $x=(u\ v)$
we have $||x||^2=||u^2||+||v||^2$ and $||\Phi x||^2 = ||Au||^2 + ||Bv||^2.$
Therefore, if $\Phi$ satisfies the RIP of order $k$ with parameter $\delta$ 
then the same is true for $A$ (take $v=0$ and $u$ $k$-sparse).
The same argument applies also to $B$.

Conversely, assume that $A$ and $B$ satisfy the RIP of order $k$
 with parameter $\delta$. Let $x=(u\ v)$ be a $k$-sparse vector.
We have 
$||\Phi x||^2 - ||x||^2 = (||\Phi u||^2 - ||u||^2)+(||\Phi v||^2 - ||v||^2).$
Both $u$ and $v$ must be $k$-sparse,
so the first term is bounded in absolute value by $\delta ||u||^2$ and the
second one by $\delta ||v||^2$. 
The result follows since $||u||^2+||v||^2 = ||x||^2$.
\end{proof}

\begin{theorem} \label{hard_rec1}
There are absolute constants $\epsilon_0,\epsilon>0$ such that 
under hypothesis~$(H_{\epsilon})$ and the choice of parameters:
$$k=n^{\frac{1}{2}+\epsilon_0}, \delta= n^{-5\epsilon/4+\epsilon^2/2}$$ 
no polynomial time algorithm can distinguish a matrix
with RIP parameter of order $k$ at most $\delta$ from a matrix with RIP
parameter of order $k$ at least $n^{-\epsilon}/4$.

Moreover, polynomial-time distinction between these two cases remains
impossible even for matrices of size $2n \times (n+N)$ where
 $N=n^{1+\epsilon_0}$. As a result, for matrices of this size the RIP
parameter of order $k$ cannot be approximated in polynomial time
within any constant factor.
\end{theorem}
The first part of the theorem follows from Theorem~\ref{highk}.
The point of Theorem~\ref{hard_rec1} is that it establishes hardness
of approximation for strictly rectangular matrices.
\begin{proof}[Proof of Theorem~\ref{hard_rec1}]
The claim on constant factor approximation follows as in Theorem~\ref{highk} 
from the relation 
$\delta=o(n^{-\epsilon}/4)$.
To prove the remainder of the theorem, 
we build on the proof of Theorem~\ref{highk}.
From a graph $G$ on $n$ vertices we construct the matrix
$$C'(G)=\left(
\begin{array}{cc}
C(G) & 0\\  0 & B_n
\end{array}\right)$$
where $C(G)$ is as in the previous section and $B_n$ is a matrix with
good restricted isometry properties. 
Its role is to ensure the rectangular format that we need for $C'(G)$.
Our specific choice for $B_n$ is the matrix constructed in~\cite{BDFKK,BDFKKb}.
It is of size $n \times N$ where $N=n^{1+\epsilon_0}$, and it satisfies the
RIP of order $n^{{\frac1{2}}+\epsilon_0}$ with parameter $n^{-\epsilon_0}$.
Moreover, $B_n$ can be constructed deterministically in time polynomial in $n$.
Note that $C'(G)$ is of size $2n \times (n+N)$ as required in the statement
of Theorem~\ref{hard_rec1}.

Choose $\epsilon$ so small that 
$(1-2\epsilon)(1-\epsilon) \geq \frac{1}{2}+\epsilon_0$ and 
${-5\epsilon/4+\epsilon^2/2} \geq -\epsilon_0.$
We thus have $\delta \geq n^{-\epsilon_0}$.
It then follows from Lemma~\ref{diag_rip}  that
$C'(G)$ satisfies the RIP of order $k$ with parameter $\delta$ if and only
if $C(G)$ does.

To complete the proof, 
let us assume that we have a distinguishing algorithm~$\cal A$ which
works for matrices of size $2n \times (n+N)$.
We use it to refute hypothesis $(H_{\epsilon})$.

If $G$ contains a clique of size $n^{1/2-{\epsilon}}$,
we saw in the proof of Theorem~\ref{highk} 
that $C(G)$ does not satisfy the RIP of order 
$n^{1/2-{\epsilon}}$ with parameter $n^{-\epsilon}/4$ 
(by Proposition~\ref{clique}).
It is {\em a fortiori} the case that this matrix does not satisfy 
the RIP of order $k=n^{\frac{1}{2}+\epsilon_0}$ 
for the same parameter value, and the same is true of $C'(G)$.

Consider now the case where $G$ is drawn from the $G(n,1/2)$ distribution.
We saw in the proof of Theorem~\ref{highk} 
that for most $G$, $C(G)$ satisfies the RIP of order 
$n^{(1-2\epsilon)(1-\epsilon)}$ with parameter $\delta$.
That order is at least as large as $k=n^{\frac{1}{2}+\epsilon_0}$,
so it is {\em a fortiori} the case that $C(G)$ satisfies the RIP of order 
$k$ with parameter $\delta$ for most $G$. As pointed out above, the same is
then true for $C'(G)$.
We can therefore refute hypothesis $(H_{\epsilon})$ 
by running algorithm $\cal A$ on~$C'(G)$.
\end{proof}
Theorem~\ref{hard_rec1} establishes hardness of approximation for an order
$k$ which is only slightly above  $n^{1/2}$. 
We can bring $k$ much closer to $n$, but for this we need a randomized version
of hypothesis $(H_{\epsilon})$:
\begin{hypothesis}[$\mathbf H_{\epsilon}'$]
There is no polynomial time {\em randomized} 
algorithm which, given as input a graph $G$
on $n$ vertices:
\begin{itemize}
\item always outputs ``yes'' 
if $G$ contains a clique of size $n^{\frac{1}{2}-\epsilon}$.

\item Outputs ``no'' with probability at least (say) 3/4 
on most graphs $G$ when $G$ is drawn from the uniform
distribution $G(n,1/2)$.
\end{itemize}
\end{hypothesis}
Note that the probability bound 3/4 in $H_{\epsilon}'$ refers to 
the {\em internal}
 coin tosses of the algorithm.
\begin{theorem} \label{hard_rec2}
Set $k=n^{(1-2\epsilon)(1-\epsilon)}$ where $\epsilon \in ]0,1/2[$.
Set also $\delta=n^{-5\epsilon/4+\epsilon^2/2}$.

Hypothesis $(H_{\epsilon}')$ implies that 
no polynomial time algorithm
can distinguish a matrix with RIP parameter of order $k$ 
at most $\delta$
from a matrix with RIP parameter of order $k$ 
at least $n^{-\epsilon}/4$. 

Moreover, polynomial-time distinction between these two cases remains
impossible even for matrices of size $2n \times 100n$.
Since $\delta=o(n^{-\epsilon}/4)$, it follows that 
for matrices of this size 
no polynomial time algorithm can approximate the RIP parameter of order $k$
within any constant factor.
\end{theorem}
\begin{proof}
As in the proof of Theorem~\ref{hard_rec1} we construct from a graph $G$
a matrix of the form
$$C'(G)=\left(
\begin{array}{cc}
C(G) & 0\\  0 & B_n
\end{array}\right).$$
For $B_n$, instead of of the deterministic construction from~\cite{BDFKK,BDFKKb} we will use a $n \times 99n$ random matrix given by Theorem~\ref{randrip}.
As before, we will certify that $G$  does
 not contain a clique of size $n^{\frac{1}{2}-\epsilon}$ if the hypothetical 
distinguishing algorithm $\cal A$ for matrices of size $2n \times 100 n$
accepts $C'(G)$. This will yield a contradiction
with Hypothesis~$(H_{\epsilon}')$.

If $G$ contains a clique of size $n^{1/2-{\epsilon}}$,
we saw in the proof of Theorem~\ref{highk} 
that $C(G)$ does not satisfy the RIP of order 
$k$ with parameter $n^{-\epsilon}/4$.
By Lemma~\ref{diag_rip}, the same is true of $C'(G)$.

Consider now the case where $G$ is drawn from the $G(n,1/2)$ distribution.
We saw in the proof of Theorem~\ref{highk} 
that for most $G$, $C(G)$ satisfies the RIP of order 
$k$ with parameter $\delta$.
As to $B_n$, note that $n \delta^2=n^{\epsilon^2-5\epsilon/2+1}$ 
and the exponent $\epsilon^2-5\epsilon/2+1 = (2-\epsilon)(1/2-\epsilon)$
is positive. Hence it follows from Theorem~\ref{randrip} that with probability approaching 1 as $n \rightarrow +\infty$, $B_n$ 
satisfies the RIP of order $k$ with parameter $\delta$.
We conclude from Lemma~\ref{diag_rip} that in this case,
$C'(G)$ satisfies the RIP of order $k$ with parameter $\delta$ for most~$G$.
\end{proof}
The constant 100 in Theorem~\ref{hard_rec2}
 can be replaced by any constant larger than~2.
Note also that the hypothetical polynomial-time algorithm in this theorem
remains deterministic: 
it is only the (hypothetical) algorithm for certifiying the absence 
of large cliques which
is randomized. 
It is clear, however, that Theorem~\ref{hard_rec2} 
can be adapted to  randomized approximation algorithms with one-sided error
(or even with two-sided error under a suitable adaptation of hypothesis 
$H'_{\epsilon}$).

The following table gives a summary of our hardness results.
They do not rule out the existence of a polynomial-time algorithm 
distinguishing between matrices with a {small}
RIP parameter and matrices with a RIP parameter larger than say 0.1.
Here {\em small} means as in Theorems~\ref{1/3th} to~\ref{hard_rec2} that the
RIP parameter goes to 0 as $n \rightarrow +\infty$. 
If convergence to 0 is not too fast then we could still
 use such a weak distinguishing algorithm for certifying most random matrices.
\newpage
\begin{table}[ht] 
\small
	\begin{center}
    \begin{tabular}{ | c | p{5cm} | c | p{3cm} | c |}
	\hline
    \multicolumn{5}{|c|}{} \\
	\multicolumn{5}{|c|}{\underline{\textbf{
Hardness Results}}} \\

    \multicolumn{5}{|c|}{} \\
	\hline\hline
	\multirow{2}{*}{} & & & & \\
	$k$  & \text{$(k,\delta_1)$ vs. $(k, \delta_2)$ - hard }  & \textbf{Result} & \textbf{Assumptions} & \textbf{Dimensions ($n\times N$)} \\
	\hline
	\hline
	\multirow{2}{*}{} & & & & \\ 
	     $n^{1/3}$ & $\delta_1 = n^{-1/4}$, $\delta_2 = n^{-1/6}/4$ 
 &  Theorem~\ref{1/3th} & $H_{1/6}$ & $n\times n$ \\
	\hline
	\multirow{2}{*}{} & & & & \\ 
	     $n^{0.6}$ & $\delta_1 = n^{-0.19}$, $\delta_2 =  n^{-1/6}/4$ &  Theorem~\ref{point6th} &   $H_{1/6}$ & $n\times n$ \\
	\hline
\multirow{2}{*}{} & & & & \\
    $n^{(1-2\epsilon)(1-\epsilon)}$ & $\delta_1 = n^{-5\epsilon/4+\epsilon^2/2}$, $\delta_2= n^{-\epsilon}/4$ &  Theorem~\ref{highk}  & $H_{\epsilon}$ & $n\times n $\\
\hline
\multirow{2}{*}{} & & & & \\
    $n^{\frac{1}{2}+\epsilon_0}$ & $\delta_1 = n^{-5\epsilon/4+\epsilon^2/2}$, $\delta_2 = n^{-\epsilon}/4$   &   Theorem~\ref{hard_rec1}  & $H_{\epsilon}$ & $2n\times (n+n^{1+\epsilon_0})$ \\
\hline
\multirow{2}{*}{} & & & & \\
   $n^{(1-2\epsilon)(1-\epsilon)}$ & $\delta_1 = n^{-5\epsilon/4+\epsilon^2/2}$, $\delta_2= n^{-\epsilon}/4$ &   Theorem~\ref{hard_rec2}  &$H'_{\epsilon}$ & 
$2n\times 100 n$ \\
\hline
\end{tabular}
\caption{We say that a matrix $\Phi$ has the $(k,\delta)$-RIP iff $(1-\delta) \leq \norm{\Phi x}^2 \leq (1+\delta)$ for every $k$-sparse unit vector $x$. By $(k,\delta_1)$ vs. $(k,\delta_2)$-hard we abbreviate the following: no polynomial time algorithm
can distinguish matrices $\Phi$ that satisfy the $(k,\delta_1)$-RIP 
from  matrices that do not satisfy the $(k,\delta_2)$-RIP. 
The absolute constant $\epsilon_0>0$ comes from~\cite{BDFKK,BDFKKb}.}
\end{center}
\end{table}


\begin{thebibliography}{10}

\bibitem{AAK07}
N.~Alon, A.~Andoni, T.~Kaufman, K.~Matulef, R.~Rubinfeld, and N.~Xie.
\newblock Testing $k$-wise and almost $k$-wise independence.
\newblock In {\em Proceedings of the Symposium on Theory of Computing (STOC)},
  pages 496--505, 2007.
\newblock available from \url{tau.ac.il/~nogaa/PDFS/aakmrx.pdf}.

\bibitem{AAMMW11}
N.~Alon, S.~Arora, R.~Manokaran, D.~Moshkovitz, and O.~Weinstein.
\newblock Inapproximability of densest $\kappa$-subgraph from average-case
  hardness.
\newblock
  \href{http://www.cs.princeton.edu/~rajsekar/papers/dks.pdf}{www.cs.princeton%
.edu/~rajsekar/papers/dks.pdf}, 2011.

\bibitem{AGHP92}
N.~Alon, O.~Goldreich, J.~Hastad, and R.~Peralta.
\newblock {Simple Construction of Almost $k$-wise Independent Random
  Variables}.
\newblock In {\em Proceedings of the Symposium on Foundations of Computer
  Science (FOCS)}, pages 544--553, 1990.

\bibitem{AKS98}
N.~Alon, M.~Krivelevich, and B.~Sudakov.
\newblock {Finding a Large Hidden Clique in a Random Graph}.
\newblock {\em Random Struct. Algorithms}, 13:457--466, 1998.

\bibitem{AKV02}
N.~Alon, M.~Krivelevich, and V.~Vu.
\newblock {On the Concentration of Eigenvalues of Random Symmetric Matrices}.
\newblock {\em Israel Journal of Mathematics}, 131:259--267, 2002.

\bibitem{BDMS12}
A.~Bandeira, E.~Dobriban, D.~Mixon, and W.~Sawin.
\newblock Certifying the restricted isometry property is hard.
\newblock \href{http://arxiv.org/abs/1204.1580}{arxiv.org/abs/1204.1580}, 2012.

\bibitem{BDFKK}
J.~Bourgain, S.~J. Dilworth, K.~Ford, S.~Konyagin, and D.~Kutzarova.
\newblock {Explicit Constructions of {RIP} Matrices and Related Problems}.
\newblock Available at~\href{http://arxiv.org/abs/1008.4535}{arxiv:1008.4535},
  August 2010.

\bibitem{BDFKKb}
J.~Bourgain, S.~J. Dilworth, K.~Ford, S.~Konyagin, and D.~Kutzarova.
\newblock {Breaking the $k^2$ Barrier for Explicit RIP Matrices}.
\newblock In {\em Proceedings of the Symposium on Theory of Computing (STOC)},
  2011.
\newblock Conference version of~\cite{BDFKK}.

\bibitem{Candes}
E.~J. Cand\`{e}s.
\newblock {The Restricted Isometry Property and its Implications for Compressed
  Sensing}.
\newblock {\em Comptes Rendus Mathematique}, 346(9-10):589 -- 592, 2008.

\bibitem{CRT06}
E.~J. Cand\`{e}s, J.~K. Romberg, and T.~Tao.
\newblock {Stable Signal Recovery from Incomplete and Inaccurate Measurements}.
\newblock {\em Communications on Pure and Applied Mathematics},
  59(8):1207--1223, 2006.

\bibitem{CandesTao}
E.~J. Cand\`{e}s and T.~Tao.
\newblock {Decoding by Linear Programming}.
\newblock {\em Information Theory, IEEE Transactions on}, 51(12):4203 -- 4215,
  2005.

\bibitem{ABG08}
A.~d'Aspremont, F.~Bach, and L.~El Ghaoui.
\newblock {Optimal Solutions for Sparse Principal Component Analysis}.
\newblock {\em J. Mach. Learn. Res.}, 9:1269--1294, 2008.

\bibitem{AspreG08}
A.~d'Aspremont and L.~El Ghaoui.
\newblock {Testing the Nullspace Property using Semidefinite Programming}.
\newblock {\em Mathematical Programming}, 127(1):123--144, 2011.
\newblock Available at~\href{http://arxiv.org/abs/0807.3520}{arxiv:0807.3520}.

\bibitem{Devore07}
R.~A. DeVore.
\newblock {Deterministic Constructions of Compressed Sensing Matrices}.
\newblock {\em J. Complex.}, 23:918--925, 2007.

\bibitem{FK81}
Z.~F\"{u}redi and J.~Koml\'{o}s.
\newblock {The Eigenvalues of Random Symmetric Matrices}.
\newblock {\em Combinatorica}, 1:233--241, 1981.

\bibitem{HK11}
E.~Hazan and R.~Krauthgamer.
\newblock How hard is it to approximate the best {Nash} equilibrium?
\newblock {\em SIAM J. Comput.}, 40:79--91, 2011.

\bibitem{IouNem08}
A.~Juditsky and A.~Nemirovski.
\newblock {On Verifiable Sufficient Conditions for Sparse Signal Recovery via
  $\ell_1$ Minimization}.
\newblock {\em Mathematical Programming}, 127(1):57--88, 2011.
\newblock Available at~\href{http://arxiv.org/abs/0809.2650}{arxiv:0809.2650}.

\bibitem{Kashin75}
B.~S. Kashin.
\newblock {The Diameters of Octahedra}.
\newblock {\em Uspekhi Mat. Nauk}, 30:251--252, 1975.

\bibitem{KZ11}
P.~Koiran and A.~Zouzias.
\newblock On the certification of the restricted isometry property.
\newblock \href{http://arxiv.org/abs/1103.4984}{arxiv.org/abs/1103.4984}, 2011.

\bibitem{PT12}
M.~Pfetsch and A.~Tillmann.
\newblock The computational complexity of the restricted isometry property, the
  nullspace property, and related concepts in compressed sensing.
\newblock \href{http://arxiv.org/abs/1205.2081}{arxiv.org/abs/1205.2081}, 2012.

\bibitem{Tao07}
T.~Tao.
\newblock Open question: deterministic {UUP} matrices.
\newblock
  \href{http://terrytao.wordpress.com/2007/07/02/open-question-deterministic-u%
up-matrices/}{terrytao.wordpress.com/2007/07/02/open-question-deterministic-uu%
p-matrices/}, 2007.

\bibitem{Vershynin}
R.~Vershynin.
\newblock {Introduction to the Non-asymptotic Analysis of Random Matrices}.
\newblock In: Compressed Sensing: Theory and Applications, eds. Y.~Eldar and
  G.~Kutyniok, pages 210-268. Cambridge University Press, 2012. Available
  at~\href{http://www-personal.umich.edu/~romanv/papers/non-asymptotic-rmt-pla%
in.pdf}{http://www-personal.umich.edu/~romanv/papers/non-asymptotic-rmt-plain.%
pdf}.

\end{thebibliography}


\section*{Appendix: Refuting $H_{\epsilon}$ for  negative $\epsilon$}

Set $k=n^{\alpha}$ where $\alpha>1/2$. In this section, we describe an algorithm
which:
\begin{itemize}
\item[(i)] always outputs ``yes'' 
if $G$ contains a clique of size $k$.

\item[(ii)] Outputs ``no clique'' 
on most graphs $G$ when $G$ is drawn from the uniform
distribution $G(n,1/2)$.
\end{itemize}
The algorithm is as follows.
\begin{enumerate}
\item Let $G$ be the input graph and $A$ its signed adjacency matrix.
Compute $\lambda_1(A)$, the largest eigenvalue of $A$.

\item Output ``yes'' if $\lambda_1(A) \geq k-1$. Otherwise, output ``no clique''.
\end{enumerate}
If $G$ contains a clique of size $k$, Lemma~\ref{expansion} shows that 
$\lambda_1(A) \geq k-1$ since $\lambda_1(A) = \sup_{||x||=1} x^TAx$
for any symmetric matrix. This algorithm therefore satisfies condition~(i).
On the other hand, for most $G$ the largest eigenvalue of $A$ is of order
$2\sqrt{n}$ by Theorem~2 in~\cite{FK81}. Since $\alpha > 1/2$, 
it follows that most $G$ satisfy the inequality $\lambda_1(A) < k-1$ 
and condition~(ii) is satisfied as well.
\end{document}